\documentclass[11pt]{article}
\usepackage{fullpage}
\usepackage{amssymb,amsmath}
\usepackage{graphicx, epsfig}

\let\idit=\id

\def\idrm#1{\ensuremath{\mathrm{#1}}}

\def\idsf#1{\ensuremath{\mathsf{#1}}}

\def\ceil#1{\lceil #1 \rceil}
\def\etal{\emph{et~al.}}

\newcommand{\no}[1]{}
\newcommand{\nono}[1]{}
\newcommand{\fullv}[1]{ }

\newtheorem{theorem}{Theorem}
\newtheorem{lemma}{Lemma}
\newtheorem{fact}{Fact}

\newtheorem{corollary}{Corollary}

\newenvironment{proof}{\trivlist\item[]\emph{Proof}:}%
{\unskip\nobreak\hskip 1em plus 1fil\nobreak$\Box$
\parfillskip=0pt%
\endtrivlist}

\newenvironment{itemize*}%
  {\begin{itemize}%
    \setlength{\itemsep}{0pt}%
    \setlength{\parskip}{0pt}%
    \setlength{\parsep}{0pt}%
    \setlength{\topsep}{0pt}%
    \setlength{\partopsep}{0pt}%
  }%
  {\end{itemize}}%

\newcommand{\cT}{{\cal T}}

\newcommand{\todo}[1]{ TO DO:  #1 \newline}

\newcommand{\locus}{\idsf{locus}}
\newcommand{\Suf}{\idsf{Suf}}
\newcommand{\RMQ}{\idsf{rmq}}
\newcommand{\occ}{\idrm{occ}}
\newcommand{\lleft}{\idit{left}}
\newcommand{\rright}{\idit{right}}

\newcommand{\eps}{\varepsilon}

\begin{document}
\title{Sorted Range Reporting\thanks{Partially funded by Millennium Nucleus Information and Coordination in
Networks ICM/FIC P10-024F, Chile.} }
\author
{
Yakov Nekrich\thanks{Department of Computer Science, University of Chile.
Email: {\tt yakov.nekrich@googlemail.com}.}
\and 
Gonzalo Navarro \thanks{Department of Computer Science, University of Chile.
Email: {\tt  gnavarro@dcc.uchile.cl}.}
}

\date{}
\maketitle
\begin{abstract}
We consider a variant of the orthogonal range reporting problem
when all points should be reported in the sorted order of their
 $x$-coordinates. We show that reporting two-dimensional points with this
 additional condition can be organized  
(almost) as efficiently as the standard range reporting.
Moreover, our results generalize and improve the previously known results 
 for the orthogonal range successor problem and can be used to obtain better 
solutions for some stringology problems. 
\end{abstract}

\section{Introduction}
\label{sec:introduction}
An orthogonal range reporting query $Q$ on a set of $d$-dimensional points 
$S$ asks for all points $p\in S$ that belong to the query 
rectangle $Q=\prod_{i=1}^d[a_i,b_i]$.  
The orthogonal range reporting problem, that is, the problem of constructing 
a data structure that supports such queries, was studied extensively; see
for example~\cite{agarwal1999geometric}.  
In this paper we consider  a variant of the two-dimensional range reporting  in which reported points must be sorted by one of their coordinates.  
Moreover, our data structures can also work in the online modus: the query answering procedure reports all points from $S\cap Q$ in 
 increasing $x$-coordinate order
until the procedure is terminated or all points in $S\cap Q$ are output.%
\footnote{We can get increasing/decreasing 
$x$/$y$-coordinate ordering via coordinate changes.}

Some simple database queries can be represented as orthogonal range reporting 
queries. For instance, identifying all company employees who are between 
20 and 40 years old and whose salary is in the range $[r_1,r_2]$ is equivalent 
to answering a range reporting query $Q= [r_1,r_2]\times [20,40]$ on a set of 
points with coordinates $(\idrm{salary}, \idrm{age})$. 
Then reporting employees with the salary-age range $Q$
sorted by their salary is equivalent to 
a sorted range reporting query. 

\tolerance=1000
Furthermore, the sorted reporting problem is a generalization of the 
orthogonal range successor problem (also known as the range next-value problem) \cite{LenhofS94,CrochemoreIR07,KKL07wads,IliopoulosCKRW08,Yu:2011}. 
The answer to an orthogonal range successor query $Q=[a,+\infty]\times [c,d]$ 
is the point with smallest $x$-coordinate\footnote{Previous works (e.g., \cite{CrochemoreIR07,Yu:2011}) use slightly different definitions, but all of them are equivalent up to a simple change of coordinate system or reduction to rank space~\cite{GabowBT84}.} among all points that are 
in the rectangle $Q$. The best previously known $O(n)$ space  data structure 
for the range 
successor queries uses $O(n)$ space and supports queries in
 $O(\log n/\log \log n)$ time~\cite{Yu:2011}. The fastest previously described 
structure  supports range successor queries in $O(\log \log n)$ 
time but needs $O(n\log n)$ space. 
\no{
A data structure that answers range successor queries in time $f(n)$ can also 
answer sorted reporting queries in $O((k+1) f(n))$ time; see the proof of Theorem~\ref{theor:spaceeff}.
Thus the linear space structure of~\cite{Yu:2011} answers queries in $O((k+1)(\log n/\log \log n))$ time.}
 In this paper we show that these results can be significantly improved.

In Section~\ref{sec:linspace} we describe two data structures for range 
successor queries.  The first structure needs $O(n)$ space and answers 
 queries in $O(\log^{\eps} n)$ time; henceforth $\eps$ denotes an arbitrarily small positive constant. The second structure  needs 
$O(n\log \log n)$ space and supports  queries in $O((\log\log n)^2)$ 
time. Both data structures can be used to answer sorted reporting queries 
in $O((k+1)\log^{\eps}n)$ and $O((k+1)(\log\log n)^2)$ time, respectively, where $k$ is the number of reported points.
In Sections~\ref{sec:3sided-rep} and~\ref{sec:2dim-range} we further improve the query time and describe a data 
structure that uses $O(n\log^{\eps} n)$ space and supports sorted reporting 
queries in $O(\log \log n + k)$ time.  As follows from the reduction of~\cite{MNSW98} and the lower bound of~\cite{PT06}, any data structure that uses $O(n\log^{O(1)}n)$ space needs 
$\Omega(\log \log n+ k)$ time to answer  (unsorted) orthogonal range reporting queries. Thus we achieve optimal query time for the sorted range reporting 
problem.  We observe that the currently best data structure 
for unsorted range reporting in optimal time~\cite{ChanLP11} also 
uses $O(n\log^{\eps} n)$ space.
In Section~\ref{sec:appl} we discuss applications of sorted reporting queries to some problems related to text indexing and some geometric problems. 
\nono{Further applications are described in the full version of this paper~\cite{NN12full}.}

Our results are valid in the word RAM model. Unless specified otherwise, we measure the space usage in words of $\log n$ 
bits. We denote by $p.x$ and $p.y$ the coordinates of a point $p$. We assume that points lie on an $n\times n$ grid, i.e., that point coordinates are integers
integers in $[1,n]$. We can reduce the more general case to this one by
reduction to rank space~\cite{GabowBT84}. The space usage will not change
and the query time would increase by an additive factor 
$pred(n)$, where $pred(n)$ is the time needed to search in a one-dimensional set of integers~\cite{BoasKZ77,PT06}.
\no{
\todo{Say that we can iterate and report only k leftmost points, k unknown in advance, say it better than in the 1st paragraph}
\todo{Say about applications,  say that we report points sorted by $x$-coordinates}
}

\section{Compact Range Trees}

The range tree is a handbook data structure frequently used for
various orthogonal range reporting problems. Its leaves contain 
the $x$-coordinates of points; a set $S(v)$ associated with each node
 $v$ contains all points whose $x$-coordinates are stored in the subtree rooted 
at $v$. We will assume that points of $S(v)$ are sorted by their $y$-coordinates.
$S(v)[i]$ will denote the $i$-th point in $S(v)$;
$S(v)[i..j]$ will denote the sorted list of  points
 $S(v)[i],S(v)[i+1],\ldots, S(v)[j]$.

A standard range tree uses $O(n\log n)$ space, but this can be reduced by 
storing compact representations of sets $S(v)$.
We will need to support the following two operations on 
compact range trees.
Given a range $[c,d]$ and a node $v$, $noderange(c,d,v)$ finds
 the range $[c_v,d_v]$  
such that $p.y\in [c,d]$ if and only if $p\in S(v)[c_v .. d_v]$  for any $p\in S(v)$.
Given an index $i$ and a  node $v$, $point(v,i)$ returns  the coordinates of 
point $S(v)[i]$. 

\begin{lemma}\cite{Chaz88,ChanLP11}\label{lemma:range}
There exists a compact range tree that uses $O(nf(n))$ space and supports 
operations $point(v,i)$ and  $noderange(c,d,v)$ in $O(g(n))$ and $O(g(n)+\log \log n)$ time, respectively, for
(i) $f(n)=O(1)$ and $g(n)=O(\log^{\eps}n)$;
(ii) $f(n)=O(\log\log n)$ and $g(n)=O(\log \log n)$;
(iii) $f(n)=O(\log^{\eps}n)$ and $g(n)=O(1)$.
\end{lemma}
\begin{proof} 
We can support $point(v,i)$ in $O(g(n))$ time using $O(nf(n))$ space as in variants $(i)$ and $(iii)$ using a result from Chazelle \cite{Chaz88}; 
we can support $point(v,i)$ in $O(\log \log n)$ time and  $O(n\log \log n)$ space using a result from Chan et al.~\cite{ChanLP11}. In the same paper \cite[Lemma 2.4]{ChanLP11}, the 
authors also showed how to support $noderange(c,d,i)$ in 
$O(g(n)+\log \log n)$ time and $O(n)$ additional space 
using a data structure that supports $point(v,i)$ in $O(g(n))$ time. 
\end{proof}

\section{Sorted Reporting in Linear Space}
\label{sec:linspace}
In this section we show how a range successor query $Q=[a,+\infty]\times [c,d]$ can be answered efficiently.
We combine the recursive approach of the van Emde Boas
structure~\cite{BoasKZ77} with compact structures for range maxima
queries. A combination of succinct range minima structures 
and range trees was also used in~\cite{ChanLP11}. A novel idea that distinguishes 
our data structure from the range reporting structure in~\cite{ChanLP11}, as well as
 from the previous range successor structures, is binary search on tree levels
 originally designed for one-dimensional searching~\cite{BoasKZ77}.
We essentially perform a  one-dimensional search for the successor 
of $a$ and answer range maxima queries at each step.
 Let $T_x$ denote the compact range tree on the
$x$-coordinates of points. $T_x$ is implemented as in variant $(i)$ of 
Lemma~\ref{lemma:range}; hence, we can find the interval $[c_v,d_v]$
for any node $v$ in $O(\log^{\eps}n)$ time. We also store a
compact structure for range maximum queries $M(v)$ in every node $v$:
given a range $[i,j]$, $M(v)$ returns the index $i\le t\le j$ of the
point $p$ with the greatest $x$-coordinate in $S(v)[i..j]$.
We also store a structure for range minimum  queries $M'(v)$.
$M(v)$ and $M'(v)$ use $O(n)$ bits and answer queries 
in $O(1)$ time \cite{Fis10}. Hence all $M(u)$ and $M'(u)$ for $u\in T_x$ use 
$O(n)$ space. 
Finally, an $O(n)$ space level ancestor structure enables us to 
find the depth-$d$ ancestor of any node $u\in T_x$ in $O(1)$ time~\cite{BenderF04}.

Let $\pi$ denote the search path for $a$ in the tree $T_x$: $\pi$
connects the root of $T_x$ with the leaf that contains the smallest
value $a_x\ge a$. Our procedure looks for the lowest node $v_f$ on
$\pi$ such that $S(v)\cap Q\not=\emptyset$.  For simplicity we assume
that the length of $\pi$ is a power of $2$.  We initialize $v_l$ to
the leaf that contains $a_x$; we initialize $v_u$ to the root
node. The node $v_f$ is found by a binary search on $\pi$.  We say
that a node $w$ is the middle node between $u$ and $v$ if $w$ is on
the path from $u$ to $v$ and the length of the path from $u$ to $w$
equals to the length of the path from $w$ to $v$. We set the node
$v_m$ to be the middle node between $v_u$ and $v_l$. Then we find the
index $t_m$ of the maximal element in $S(v_m)[c_{v_m}..d_{v_m}]$ and the
point $p_m=S(v_m)[t_m]$.  If $p_m.x\ge a$, then $v_f$ is either $v_m$
or its descendant; hence, we set $v_u=v_m$.  If $p_m.x<a$, then $v_f$ is
an ancestor of $v_m$; hence, we set $v_l=v_m$. The search procedure
continues until $v_u$ is the parent of $v_m$.  Finally, we test nodes
$v_u$ and $v_l$ and identify $v_f$ (if such $v_f$ exists).
\begin{fact}
  If the child $v'$ of $v_f$ belongs to $\pi$, then $v'$ is the left
  child of $v_f$.
\end{fact}
\begin{proof}
  Suppose that $v'$ is the right child of $v_f$ and let $v''$ be the
  sibling of $v'$. By definition of $v_f$, $Q\cap
  S(v')=\emptyset$. Since $v'$ belongs to $\pi$ and $v''$ is the left
  child, $p.x<a$ for all points $p\in S(v'')$.  Since
  $S(v_f)=S(v')\cup S(v'')$, $Q\cap S(v_f)=\emptyset$ and we obtain a
  contradiction.
\end{proof}
Since $v'\in \pi$ is the left child of $v_f$, $p.x\ge a$ for all $p\in
S(v'')$ for the sibling $v''$ of $v$.  Moreover, $p.x< a$ for all
points $p\in S(v')[c_{v'},d_{v'}]$ by definition of $v_f$.  Therefore
the range successor is the point with minimal $x$-coordinate in
$S(v'')[c_{v''}..d_{v''}]$.

The search procedure visits $O(\log \log n)$ nodes and spends
$O(\log^{\eps}n)$ time in each node, thus the total query time is
$O(\log^{\eps}n\log\log n)$. By replacing $\eps'<\eps$ in the above
construction, we obtain the following result.
\begin{lemma}\label{lemma:linsucc}
  There exists a data structure that uses $O(n)$ space and answers
  orthogonal range successor queries in $O(\log^{\eps}n)$ time.
\end{lemma}
If we use the compact tree that needs $\Theta(n\log \log n)$ space,
then $g(n)=\log \log n$. Using the same structure as in the proof of 
Lemma~\ref{lemma:linsucc}, we obtain the following.
\begin{lemma}\label{lemma:llsucc}
  There exists a data structure that uses $O(n\log\log n)$ space and
  answers orthogonal range successor queries in $O((\log \log n)^2)$
  time.
\end{lemma}

\paragraph{Sorted Reporting Queries.} 
We can answer sorted reporting queries by answering a sequence of 
range successor queries. Consider a query $Q=[a,b]\times [c,d]$. 
Let $p_1$ be the answer to the range successor query $Q_1=[a,+\infty]\times [c,d]$.
For $i\ge 2$, let $p_i$ be the answer to the query 
$Q_i=[p_{i-1}.x,+\infty]\times [c,d]$.  
The sequence of points $p_1,\ldots p_k$ is the sequence of $k$ leftmost 
points in $[a,b]\times [c,d]$ sorted by their $x$-coordinates. We observe 
that our procedure also works in the \emph{online modus} when  $k$ is not known in advance. That is, we can output the points of $Q\cap S$ in the left-to-right order until the procedure is stopped by the user or 
all points in $Q\cap S$ are reported.
\no{
 leftmost point,  then the next one, etc;   the query can be stopped at any time by 
the user. The total query time is the same as the time needed to answer 
$k$ range successor queries. } 
\begin{theorem}\label{theor:spaceeff}
There exist a data structures that uses $O(n)$ space and answer
  sorted range reporting queries in $O((k+1)\log^{\eps}n)$ time, and that
use $O(n\log \log n)$ space 
and answer those queries in $O((k+1)(\log \log n)^2)$ 
time.
\end{theorem}

\section{Three-Sided Reporting in Optimal Time}
\label{sec:3sided-rep}
In this section we present optimal time data structures for two special 
cases of sorted two-dimensional queries.
In the first part of this section we describe a data structure that
answers sorted one-sided queries: for a query $c$ we report all points
$p$, $p.y\le c$, sorted in increasing order of their $x$-coordinates. 
Then we will show how to answer three-sided queries, i.e., to report all points $p$, $a\le p.x\le b$ and
$p.y\le c$, sorted in increasing order by their $x$-coordinates.

\paragraph{One-Sided Sorted Reporting.} We start by describing a data
structure that answers queries in $O(\log n+k)$ time; our solution is  based on a standard range tree
 decomposition of the query interval $[1,c]$ into $O(\log n)$ intervals. Then we show how
to reduce the query time to $O(k+\log \log n)$. This improvement uses an additional data structure 
for the case when $k\le \log n$ points must be reported. 

We construct a range tree on the
$y$-coordinates. For every node $v\in T$, the list $L(v)$ contains all
points that belong to $v$ sorted by their $x$-coordinates. Suppose
that we want to return $k$ points $p$ with smallest $x$-coordinates
such that $p.y\leq c$.  We can represent the interval $[1,c]$ as a
union of $O(\log n)$ node ranges for nodes $v_i\in T$.  The search
procedure visits each $v_i$ and finds the leftmost point (that is, the
first point) in every list $L(v_i)$.  Those points are kept in a data
structure $D$. Then we repeat the following step $k$ times: We find the 
leftmost point $p$ stored in $D$, output $p$ and remove it from
$D$. If $p$ belongs to a list $L(v_i)$, we find the point $p'$ that
follows $p$ in $L(v_i)$ and insert $p'$ into $D$.  As $D$ contains
$O(\log n)$ points, we support updates and find the leftmost point
in $D$ in $O(1)$ time \cite{FW94}.  Hence, we can initialize $D$ in
$O(\log n)$ time and then report $k$ points in $O(k)$ time.

We can reduce the query time to $O(k+\log \log n)$ by constructing additional data
structures.  If $k\geq \log n$ the data structure described above
already answers a query in $O(k+\log n)=O(k)$ time.  The case $k\leq
\log n$ can be handled as follows.  We store for each $p\in S$ a list
$V(p)$. Among all points $p'\in S$ such that $p'.y\leq p.y$ the list
$V(p)$ contains $\log n$ points with the smallest
$x$-coordinates. Points in $V(p)$ are sorted in increasing order by
their $x$-coordinates.  To find $k$ leftmost points in $[1,c]$ for $k
<\log n$, we identify the highest point $p_c\in S$ such that $p_c.y \le
c$ and report the first $k$ points in $V(p_c)$. The point $p_c$ can be
found in $O(\log \log n)$ time using the van Emde Boas data structure~\cite{BoasKZ77}.  If $p_c$ is known, then a query can be answered in $O(k)$ time for any 
value of $k$.

One last improvement will be important for the data structure of Lemma~\ref{lemma:3sid}. Let $S_m$ denote the set of $\ceil{\log\log n}$ lowest points in $S$. We store the $y$-coordinates of $p\in S_m$
in the $q$-heap F. Using $F$, we can find the highest $p_m\in S_m$, 
such that $p_m.y\le c$, in $O(1)$ time \cite{FW94}. 
\no{Suppose that we want to report $k\ge \log \log n$ leftmost points $p\in S$, $p\leq a$, in sorted order.}   
Let $n_c=|\{\,p\in S\,|\,p.y\le c\,\}|$. If $n_c\le \log\log n$, 
then $p_m=p_c$. As described above, we can answer a query in $O(k)$ time 
when $p_c$ is known. Hence, a query can be answered in $O(k)$ time if $n_c\le \log \log n$.
\no{If $n_a>\log \log n$, then we can answer a query in $O(k+\log \log n)=O(k)$ time.} 
\begin{lemma}\label{lemma:1sid}
  There exists an $O(n\log n)$ space data structure that supports
  one-sided sorted range 
  reporting queries in $O(\log \log n + k)$ time. If the highest point
  $p$ with $p.y\le c$ is known, then one-sided sorted queries can be
  supported in $O(k)$ time. 
  If $|\{\,p\in S\,|\,p.y\le c\,\}|\le \log \log n$, a sorted range 
  reporting query $[1,c]$ can be answered in $O(k)$ time.
\no{If $k\geq \log \log n$, then  one-sided  sorted  queries  can be answered in $O(k)$ time.}
\no{ where $n_a=|\{\,p\in S\,|\,p.y\le c\,\}|$.}
\end{lemma}

\paragraph{Three-Sided Sorted Queries.} We construct a range tree on
$x$-coordinates of points. For any node $v$, the data structure $D(v)$ of Lemma~\ref{lemma:1sid}
supports one-sided queries on $S(v)$ as described above. 
\no{We also store $\ceil{\log \log n}$ lowest points from $S(v)$ 
in a data structure $L(v)$.  Using the result of~\cite{FW94}, 
we can find for any $c$
the highest $p\in L(v)$ such that $p\le c$. 
Hence, we can use $L(v)$ to determine whether $S(v)$ contains at least 
$\log \log n$ points that are below $c$. 
If this is not the case, we can find the highest $p\in S(v)$, $p.y\le c$, 
using $L(v)$. 
}
For each
root-to-leaf path $\pi$ we store two data structures, $R_1(\pi)$ and
$R_2(\pi)$.  Let $\pi^+$ and $\pi^-$ be defined as follows. If $v$
belongs to a path $\pi$ and $v$ is the left child of its parent, then
its sibling $v'$ belongs to $\pi^+$.  If $v$ belongs to $\pi$ and $v$
is the right child of its parent, then its sibling $v'$ belongs to
$\pi^-$.  The data structure $R_1(\pi)$ contains the lowest point in
$S(v')$ for each $v'\in \pi^+$; if $v\in \pi$ is a leaf, $R_1(\pi)$
also contains the point stored in $v$.  The data structure $R_2(\pi)$
contains the lowest point in $S(v')$ for each $v'\in \pi^-$; if $v\in
\pi$ is a leaf, $R_2(\pi)$ also contains the point stored in $v$.  Let
$lev(v)$ denote the level of a node $v$ (the level of a node $v$ is the 
length of the path from the root to $v$). If a point $p\in R_i(\pi)$, $i=1,2$, comes
from a node $v$, then $lev(p)=lev(v)$.  For a given query $(c,l)$ the
data structure $R_1(\pi)$ ($R_2(\pi)$) reports points $p$ such that
$p.y\le c$ and $lev(p)\ge l$ sorted in decreasing (increasing) order
by $lev(p)$.  Since a data structure $R_i(\pi)$, $i=1,2$, contain
$O(\log n)$ points, the point with the $k$-th largest (smallest) value of $lev(p)$
among all $p$ with $p.y\le c$ can be found in $O(1)$ time. 
\no{first $k$ points can be found in $O(k)$ time.}
The implementation of structures $R_i(\pi)$ is based on standard bit 
techniques and will be described in the full version.

Consider a query $Q=[a,b]\times [1,c]$.  Let $\pi_a$ and $\pi_b$ be
the paths from the root to $a$ and $b$ respectively. Suppose that the
lowest node $v\in \pi_a\cap\pi_b$ is situated on level $lev(v)=l$. Then
all points $p$ such that $p.x\in [a,b]$ belong to some node $v$ such
that $v\in \pi_a^+$ and $lev(v)> l$ or $v\in \pi^-_b$ and $lev(v)>l$.
We start by finding the leftmost point $p$ in $R_1(\pi_a)$ such that
$lev(p)>l$ and $p.y\leq c$. 
Since the $x$-coordinates of points in
$R_1(\pi_a)$ decrease as $lev(p)$ increases, this is equivalent to
finding the point $p_1\in R_1(\pi_a)$ such that $p_1.y \leq c$ and $lev(p_1)$ is maximal.
If $lev(p_1)>l$, we visit the node $v_1\in \pi^+_a$ that contains $p_1$; \no{We find the highest point  $p_l\in L(v_1)$, $p_l.y\le c$. If $p_l$ is the last (highest) point in $L(v_1)$, then $S(v_1)$  
contains at least $\log \log n$ points $p$, such that $p.y\le c$. Otherwise $p_l$ is the highest among all $p\in S(v_1)$ 
such that $p\le c$.
Then, we use $D(v_l)$ to report the $k$ leftmost points $p'\in S(v)$ such that $p'.y\le c$. 

If $p_1$ is not the highest point in $L(v_1)$, then $p.y \le p_l.y$ for any $p\in S(v_1)$ such that $p.y\le c$. 
}
 using $D(v_1)$, we report the $k$ leftmost points $p'\in S(v_1)$ such that $p'.y\le c$. 
Then, we find the point $p_2$ with the next largest value of 
$lev(p)$ among all $p\in R_1(\pi_a)$ such that $p.y\le c$; we visit the node 
$v_2\in \pi^+_a$ that contains $p_2$ and proceed 
as above. 
The procedure continues until $k$ points are output or there are no 
more points $p\in R_1(\pi_a)$, $lev(p)>l$ and $p.y\le c$. 
If $k'<k$ points were reported, we visit selected nodes $u\in \pi^-_b$ and report remaining $k-k'$ points using a symmetric procedure. 

Let $k_i$ denote the number of reported points from the set $S(v_i)$ 
and let $m_i=Q\cap S(v_i)$. 
We spend $O(k_i)$ time in a visited node $v_i$ if $k_i\ge \log \log n$
or $m_i< \log \log n$. If  $k_j<\log\log n$ and $m_j\ge \log \log n$, then we spend 
$O(\log \log n +k_j)$ time in the respective node $v_j$. 
Thus we spend $O(\log \log n + k_j)$ time in a node $v_j$
only if $m_j>k_j$, i.e., only if not all points from $S(v_j)\cap Q$
are reported.
Since at most one such 
node $v_j$ is visited, the total time needed to answer all one-sided queries is $O(\sum_i k_i + \log \log n)=O(\log \log n + k)$. 
\no{
in the decreasing order of $lev(p)$. For every found point $p\in
R_1(\pi_a)$, we visit the node $v\in \pi^+_a$ that contains $v$.
Using $D(v)$, we report the $k$ leftmost points $p'\in S(v)$ such that
$p'.y\le c$. 
}
\begin{lemma}\label{lemma:3sid}
  There exists an $O(n\log^2 n)$ space data structure that answers
  three-sided sorted reporting queries in $O(\log \log n + k)$ time.
\end{lemma}

\paragraph{Online queries.}
We assumed in Lemmas~\ref{lemma:1sid} and~\ref{lemma:3sid} that 
parameter $k$ is fixed and given with the query. 
Our data structures can also support queries in the online modus using 
the method originally described in~\cite{BrodalFGL09}. 
The main idea is that we find roughly $\Theta(k_i)$ leftmost points from the 
query range for $k_i=2^i$ and $i=1,2,\ldots$; while $k_i$ points are reported, we simultaneously compute the following $\Theta(k_{i+1})$ points in the background. 
For a more extensive description, refer to \cite[Section 4.1]{NavN12}, where the same method for a slightly different problem is described.

\section{Two-Dimensional Range Reporting in Optimal Time}
\label{sec:2dim-range}
We store points in a compact range tree $T_y$ on
$y$-coordinates. We use the variant $(iii)$ of Lemma~\ref{lemma:range} that uses 
$O(n\log^{\eps}n)$ space and retrieves the coordinates of the $r$-th point
 from
 $S(v)$ in $O(1)$ time. Moreover, the sets $S(v)$, $v\in T_y$, are divided into
groups $G_i(v)$. Each $G_i(v)$, except of the last one, contains
$\lceil\log^3 n\rceil$ points. For $i<j$, each point assigned to $G_i(v)$ has smaller $x$-coordinate than any point in $G_j(v)$. The set $S'(v)$ contains selected elements from $S(v)$. If $v$ is the right child of its
parent, then $S'(v)$ contains $\ceil{\log\log n}$ points with smallest
$y$-coordinates from each group $G_i(v)$; structure $D'(v)$
supports three-sided sorted queries of the form $[a,b]\times [0,c]$ on
points of $S'(v)$.  If $v$ is the left child of its parent, then
$S'(v)$ contains $\ceil{\log\log n}$ points with largest $y$-coordinates from
each group $G_i(v)$; data structure $D'(v)$ supports three-sided
sorted queries of the form $[a,b]\times [c,+\infty]$ on points of
$S'(v)$.  For each point $p'\in S'(v)$ we store the index $i$ of the
group $G_i(v)$ that contains $p$.  We also store the point with
the largest $x$-coordinate from each $G_i(v)$ in a structure
$E(v)$ that supports $O(\log \log n)$ time searches \cite{BoasKZ77}.

For all points in each group $G_i(v)$ we store an array $A_i(v)$ that
contains  points sorted by their $y$-coordinates.  Each point is
specified by the rank of its $x$-coordinate in $G_i(v)$; so each
entry uses $O(\log \log n)$ bits of space.
\no{
We answer a two-dimensional query by answering two three-sided 
queries in the online modus. Data structures $D'(u)$ are main 
tools for answering relevant three-sided queries. We also use 
the structures for single groups 
}

To answer a query $Q=[a,b]\times [c,d]$, we find the lowest common
ancestor $v_c$ of the leaves that contain $c$ and $d$.  Let $v_l$ and
$v_r$ be the left and the right children of $v_c$.  All points in
$Q\cap S$ belong to either $([a,b]\times [c,+\infty])\cap S(v_l)$ or
$([a,b]\times [0,d])\cap S(v_r)$.  We generate the sorted list of $k$
leftmost points in $Q\cap S$ by merging the lists of $k$ leftmost
points in $([a,b]\times [c,+\infty])\cap S(v_l)$ and $([a,b]\times
[0,d])\cap S(v_r)$. Thus it suffices to answer sorted three-sided
queries $([a,b]\times [c,+\infty])$ and $([a,b]\times [0,d])$ in nodes
$v_l$ and $v_r$ respectively.

We consider a query $([a,b]\times [0,d])\cap S(v_r)$; query $[a,b]\times [c,+\infty]$ is answered symmetrically. Assume
$[a,b]$ fits into one group $G_i(v_r)$, i.e., all points $p$ such that
$a\le p.x\le b$ belong to one group $G_i(v_r)$.  We can find the
$y$-rank $d_r$ of the highest point $p\in G_i(v_r)$, such that $p.y
\leq d$ in $O(\lg \lg n)$ time by binary search in $A_i(v_r)$.  Let $a_r$
and $b_r$ be the ranks of $a$ and $b$ in $G_i(v_r)$.  We can find the
positions of $k$ leftmost points in $([a_r,b_r]\times [0,d_r])\cap
G_i(v_r)$ using a data structure $H_i(v_r)$.
$H_i(v_r)$ contains the $y$-ranks and $x$-ranks of points in $G_i(v_r)$ 
and answers sorted three-sided queries on $G_i(v_r)$. 
By Lemma~\ref{lemma:3sid}, $H_i(v_r)$ uses $O(|G_i(v_r)|(\log \log n)^3)$ bits 
and supports  queries in $O(\log \log \log n + k )$ time.
 Actual coordinates of points can be obtained from their ranks in $G_i(v_r)$ in  $O(1)$ time per point: if the $x$-rank of a point 
is known, we can compute its position in $S(v_r)$; we obtain 
$x$-coordinates of the $i$-th point in $S(v_r)$ using variant 
$(iii)$ of Lemma~\ref{lemma:range}.

Now assume $[a,b]$ spans several groups $G_i(v_r),\ldots,G_j(v_r)$
for $i<j$. That is, the $x$-coordinates of all points in groups
$G_{i+1}(v_r),\ldots,G_{j-1}(v_r)$ belong to $[a,b]$; the
$x$-coordinate of at least one point in $G_i(v_r)$ ($G_j(v_r)$) is
smaller than $a$ (larger than $b$) but the $x$-coordinate of at least
one point in $G_i(v_r)$ and $G_j(v_r)$ belongs to $[a,b]$. Indices $i$ and 
$j$ are found in $O(\log \log n)$ time using $E(v_r)$.  We report
at most $k$ leftmost points in $([a,b]\times [0,d])\cap G_i(v_r)$
just as described above.

Let $k_1=|([a,b]\times [0,d])\cap G_i(v_r)|$; if $k_1\ge k$, the query
is answered. Otherwise, we report $k'=k-k_1$ leftmost points in
$([a,b]\times [0,d])\cap (G_{i+1}(v_r)\cup\ldots\cup G_{j-1}(v_r))$
using the following method.  Let $a'$ and $b'$ be the minimal and
the maximal $x$-coordinates of points in $G_{i+1}(v_r)$ and
$G_{j-1}(v_r)$, respectively.  The main idea is to answer the query
$Q'=([a',b']\times [0,d])\cap S'(v_r)$ in the online modus using the data structure
$D'(v_r)$. If some group $G_t(v_r)$, $i<t<j$, contains less than $\ceil{\log
\log n}$ points $p$ with $p.y\le d$, then all such $p$ belong to
$S'(v_r)$ and will be reported by $D'(v_r)$.  Suppose that $D'(v_r)$
reported $\log \log n$ points that belong to the same group
$G_t(v_r)$. Then we find the rank $d_t$ of $d$ among the
$y$-coordinates of points in $G_t(v_r)$. Using $H_t(v_r)$, we report
the positions of all points $p\in G_t(v_r)$, such that the rank of
$p.y$ in $G_t(v_r)$ is at most $d_t$, in the left-to right order; we can
also identify the coordinates of every such $p$ in $O(1)$ time per point.
The query to $H_t(v_r)$ is terminated when all such points are
reported or when the total number of reported points is $k$.

We need $O(\log \log n +k_t)$ time to answer a query on $H_t(v_r)$, where 
$k_t$ denotes the number of reported points from $G_t(v_r)$.  
Let $m_t=|Q'\cap G_t(v_r)|$
If $G_t$ is the last examined group, then $k_t\le m_t$; otherwise 
$k_t=m_t$.
We send a query to $G_t(v_r)$ only if $G_t(v_r)$ contains at least 
$\log \log n$ points from  $Q'$. 
Hence,  a query to $G_t(v_r)$ takes $O(\log \log n + k_t)=O(k_t)$ time, 
unless $G_t(v_r)$ is the last examined group. 
Thus all necessary queries to $G_t(v_r)$ for $i{<}t{<}j$ take $O(\log \log n + k)$ time.

Finally, if the total number of points in $([a,b]\times [0,d])\cap (G_i(v_r)\cup\ldots\cup G_{j-1}(v_r))$ is smaller than $k$, we also report the remaining points from $([a,b]\times [0,d])\cap G_j(v_r)$. 

\tolerance=1500
The compact tree $T_y$ uses $O(n\log^{\eps}n)$ words of space.  
A data structure $D'(v)$ uses $O(|S'(v)|\log^2 n\log\log n)= O(|S(v)|\log \log n/\log n)$ words 
of space. Since all sets $S(v)$, $v\in T_y$, contain $O(n\log n)$ 
points, all $D'(v)$ use $O(n\log \log n)$ words of space.  
A data structure for a group $G_i(v)$  uses $O(|G_i(v)|(\log \log n)^3)$
 bits. Since all $G_i(v)$ for all $v\in T_y$ contain $O(n\log n)$ elements, data structures for all groups $G_i(v)$ use $O(n(\log \log n)^3)$ 
words of $\log n$ bits. 
\begin{theorem}\label{theor:2dim}
  There exists a $O(n\log^{\eps} n)$ space data structure that answers
  two-dimensional sorted reporting queries in $O(\log \log n + k)$ time.
\end{theorem}

\section{Applications}
\label{sec:appl}
\tolerance=1000
In this section we will describe data structures for several indexing
 and computational geometry problems. A text (string) $T$ of length $n$ 
 is pre-processed and stored in a data structure so that 
certain queries concerning some substrings of $T$ can be answered 
efficiently.
\paragraph{Preliminaries.}
In a suffix tree $\cT$ for a text $T$, every leaf of $\cT$ is associated with a suffix of 
$T$. If the leaves of $\cT$ are listed from left to right, then the 
corresponding suffixes of $T$ are lexicographically sorted. 
For any pattern $P$, we can find in $O(|P|)$ time the special 
node $v\in \cT$, called the \emph{locus} of $P$. 
The starting position of every suffix in the subtree of $v=\locus(P)$ is the 
location of an occurrence of $P$. 
We define the rank of a suffix $\Suf$ as the number of $T$'s suffixes 
that are lexicographically smaller than or equal to $\Suf$.  
The ranks of all suffixes in $v=\locus(P)$ belong to an interval 
$[\lleft(P),\rright(P)]$, where $\lleft(P)$ and $\rright(P)$ denote the ranks 
of the leftmost and the rightmost suffixes in the subtree of $v$. 
Thus for any pattern $P$ there is a unique range $[\lleft(P),\rright(P)]$; 
pattern $P$ occurs at position $i$ in $T$ if and only if 
the rank of suffix $T[i..n]$ belongs to $[\lleft(P),\rright(P)]$.  
Refer to~\cite{Gusfield1997} for a more extensive description of
suffix trees and related concepts. 

We will frequently use a special set of points, further called 
\emph{the position set for $T$}.
Every point $p$ in the position set corresponds to a unique suffix $\Suf$
of a string $T$; 
the $y$-coordinate of $p$ equals to the rank of $\Suf$ and the $x$-coordinate 
of $p$ equals to the starting position of $\Suf$ in $T$.

\paragraph{Successive List Indexing.}
In this problem a query consists of a pattern $P$ 
and an index $j$, $1\le j\le n$. We want to find the first (leftmost) 
occurrence of $P$  at position $i\ge j$. 
A successive list indexing query $(P,j)$ is equivalent to finding the 
point $p$ from the position set such that $p$ belongs to the range  $[j,n]\times [left(P),right(P)]$ and the $x$-coordinate of $p$ is minimal. 
Thus a list indexing query is equivalent to a range successor query on 
the position set. Using Theorems~\ref{theor:spaceeff} and~\ref{theor:2dim} 
to answer range successor queries, we obtain the following result.
\begin{corollary}\label{cor:succind}
We can store a string $T$ in an $O(nf(n))$ space data structure, so that 
for any pattern $P$ and any index $j$, $1\le j\le n$, the leftmost 
occurrence of $P$ at position $i\ge j$ can be found in $O(g(n))$ time 
for 
 (i) $f(n)=O(1)$ and $g(n)=O(\log^{\eps}n)$;
(ii) $f(n)=O(\log\log n)$ and $g(n)=O((\log \log n)^2)$;
(iii) $f(n)=O(\log^{\eps}n)$ and $g(n)=O(\log \log n)$.
\end{corollary}

\paragraph{Range Non-Overlapping Indexing.}
In the string statistics problem we want to find the maximum number of non-overlapping occurrences of a pattern $P$. In~\cite{KKL07wads} the 
\emph{range non-overlapping indexing problem} was introduced: 
instead of just computing the maximum number of occurrences 
we want to find the longest sequence of non-overlapping occurrences of $P$.
It was shown \cite{KKL07wads} that the range non-overlapping 
indexing problem can be solved via $k$ successive list indexing 
queries; here $k$ denotes the maximal number of non-overlapping occurrences. 

\begin{corollary}\label{cor:rangenon}
The range non-overlapping indexing problem can be solved  in $O(|P|+kg(n))$ time with an $O(n f(n))$ space data structure, where $g(n)$ and $f(n)$ 
are defined as in Corollary~\ref{cor:succind}.
\end{corollary}

Other, more far-fetched applications, are described next.

\subsection{Pattern Matching with Variable-Length Don't Cares}
\label{sec:regpat}
We must determine whether a query pattern $P=P_1*P_2*P_3\ldots *P_m$ 
occurs in $T$. The special symbol $*$ is the Kleene star
 symbol; it corresponds to an arbitrary sequence of (zero or more) 
characters from the original alphabet of $T$.
The parameter $m$ can be specified at query time. 
In~\cite{YuWK10} the authors showed how to answer such queries in 
$O(\sum_{i=1}^m |P_i|)$ and $O(n)$ space in the case when the alphabet size 
is $\log^{O(1)}n$. In this paper we describe  a data structure 
for an arbitrarily large alphabet.
Using the approach of~\cite{YuWK10}, we can reduce such a query for  $P$ 
to answering $m$ successive list indexing queries.  
First, we identify the leftmost occurrence of $P_1$ in $T$ by answering 
the successive list indexing query $(P_1,1)$. 
Let $j_1$ denote the leftmost position of $P_1$.  
$P_1*P_2*P_3\ldots *P_m$ occurs in $T$ if and only if $P_2*P_3\ldots *P_m$ occurs at position $i\ge j_1+ |P_1|$. We find the leftmost occurrence
 $j_2\ge j_1+|P_1|$ of $P_2$ by answering the query $(P_2,j_1+|P_1|)$.  
$P_2*P_3\ldots *P_m$ occurs in $T$ at position $i_2\ge j_1+ |P_1|$ 
if and only if $P_3*P_m$ occurs at position $i_3\ge j_2+ |P_2|$. 
Proceeding in the same way we find the leftmost possible positions 
for $P_4*\ldots *P_m$. Thus we answer $m$ successive list indexing 
queries $(P_t,i_t)$, $t=1,\ldots, m$; here $i_1=1$, $i_t=j_{t-1}+|P_{i-1}|$ 
for $t\ge 2$, and $j_{t-1}$ denotes the answer to the $(t-1)$-th query.
\begin{corollary}\label{cor:dontcares}
We can determine whether a text $T$ contains a substring $P=P_1*\ldots P_{m-1}*P_m$   in $O(\sum_{i=1}^m|P_i|+mg(n))$ time using an $O(n f(n))$ space data structure, where $g(n)$ and $f(n)$ are defined as in Corollaries~\ref{cor:succind} and~\ref{cor:rangenon}.
\end{corollary}
\nono{Further applications of sorted range reporting and orthogonal range successor queries are discussed in the full version of this 
paper~\cite{NN12full}.}

\subsection{Ordered Substring Searching}
\label{sec:sortsubstr}
Suppose that a data structure contains a text $T$ and we want to report 
occurrences of a query pattern $P$ in the left-to-right order, i.e., in the 
same order as they appear in $T$; in some case we may want to find only 
the $k$ leftmost occurrences.
 In this section we describe two solutions for this problem. Then we show how sorted range reporting 
can be used to solve the position-restricted variant of this problem. We denote by $\occ$ the number of $P$'s occurrences in $T$ that are reported when a query is answered.

\paragraph{Data Structure with Optimal Query Time.}
Such queries can be answered in 
$O(|P|+\occ)$ time and $O(n)$ space using the suffix tree and the data structure  of Brodal \etal~\cite{BrodalFGL09}. 
Positions of suffixes are stored in lexicographic order in the suffix array $A$;
the $k$-th entry $A[k]$ contains the starting position of the $k$-th suffix in 
the lexicographic order.  In~\cite{BrodalFGL09} the authors described an $O(n)$ space data structure that answers online sorted range reporting queries: 
for any $i\ge j$, we can report in $O(j-i+1)$ time 
all entries $A[t]$, $i\le t\le j$, sorted in increasing order by their 
values. Occurrences of a pattern $P$ can be reported in the left-to-right order as follows.  Using a suffix tree, we find $\lleft(P)$ and $\rright(P)$ in $O(|P|)$ time. Then we report all suffixes in the interval 
$[\lleft(P),\rright(P)]$ sorted by their starting positions 
using the data structure of~\cite{BrodalFGL09} on $A$. 
\begin{corollary}\label{cor:sortlin}
We can answer a sorted  substring matching query in $O(|P|+\occ)$ time using a $O(n)$ space data structure
\end{corollary}

\paragraph{Succinct Data Structure.}
The space usage of a data structure for sorted pattern matching can be 
further reduced.  We store a compressed suffix array for $T$ and a 
succinct data structure for range minimum queries. 
We use the implementation of the compressed suffix array described in~\cite{GrossiGV03} that needs $(1+1/\eps)nH_k+o(n)$ bits for $\sigma=\log^{O(1)}n$, where $\sigma$ denotes  the alphabet size  and $H_k$ is the $k$-th order entropy.
Using the results of~\cite{GrossiGV03}, we can find the position of the $i$-th lexicographically smallest  suffix in $O(\log^{\eps} n)$ time.  
We can also find $\lleft(P)$ and $\rright(P)$ for any $P$ in 
$O(|P|)$ time. 
We also store the range minimum data structure~\cite{Fis10} for the 
array $A$ defined above.  For any $i\le j$, we can find  such $k=\RMQ(i,j)$ that 
$A[k]\le A[t]$ for any $i\le t\le j$. Observe that $A$ itself is not stored; we only store the structure from~\cite{Fis10} that uses $O(n)$ bits of space. 
Now occurrences of $P$ are reported as follows. 
An initially empty queue $Q$ contains suffix positions;
with every suffix position $p$ we also store an interval $[l_p,r_p]$ and the rank $i_p$ of the suffix that starts at 
position $p$. Let $l=\lleft(P)$ and $r=\rright(P)$. We find 
$i_f=\RMQ(l,r)$ and  the position $p_f$ of the suffix with 
rank $i_f$. The position $p_f$ with its rank $i_f$ and the associated interval $[l,r]$ is inserted into $Q$. We repeat the following steps until $Q$ is empty.
The item with the minimal value of $p_t$ is extracted from 
$Q$. Let $i_t$ and $[l_t,r_t]$ denote the rank and interval stored with $p_t$. We answer queries $i'=\RMQ(l_t,i_t-1)$ and 
$i''=\RMQ(i_t+1,r_t)$ and identify the positions $p'$, $p''$ of suffixes with ranks $i'$, $i''$. Finally, we insert items 
$(p',i',[l_t,i_t-1])$ and $(p'',i'',[i_t+1,r_t])$ into 
$Q$. Using the van Emde Boas data structure, we can 
implement each operation on $Q$ in $O(\log \log n)$ time. 
We can find the position of a suffix with rank $i$ in 
$O(\log^{\eps} n)$ time. Thus the total time that 
we need to answer a query is $O(|P|+\occ\log^{\eps} n)$. 
Our data structure uses $(1+1/\eps)nH_k+O(n)$  bits. 
We observe however that we need $O(\occ\log n)$ additional 
bits at the time when a query is answered. 
\begin{corollary}\label{cor:sortcomp}
If the alphabet size $\sigma=\log^{O(1)}n$, then 
we can answer an ordered substring searching   query in $O(|P|+\occ\log^{\eps}n)$ time 
using a $(1+1/\eps)nH_k+O(n)$-bit data structure
\end{corollary}

\paragraph{Position-Restricted Ordered Substring Searching}
The position restricted substring searching problem was introduced by  M{\"a}kinen and Navarro in~\cite{MakinN07}. Given a range $[i,j]$ 
we want to report all occurrences of $P$ that start at 
position $t$, $i\le t\le j$. 
If we want to report occurrences of $P$ at positions from 
$[i,j]$ in the sorted order, then this is equivalent to 
answering a sorted range reporting query $[i,j]\times [\lleft(P),\rright(P)]$. Hence, we can obtain the same 
time-space trade-offs as in Theorems~\ref{theor:spaceeff} and~\ref{theor:2dim}.

\subsection{Maximal  Points in a 2D Range and Rectangular Visibility}

A point $p$ \emph{dominates} another point $q$ if $p.x\ge q.x$ 
and $p.y\ge q.y$. A point $p\in S$ is \emph{a maximal point} 
if $p$ is not dominated by any other point $q\in S$. 
In a two-dimensional maximal points range  query, we must find  
all maximal points in  $Q\cap S$ for a query rectangle $Q$.
We refer to~\cite{BT11} and references therein for description of 
previous results.

We can answer such  queries using orthogonal range successor queries. 
For simplicity, we assume that all points have different $x$- and $y$-coordinates.
Suppose that maximal points in the range 
$Q=[a,b]\times [c,d]$ must be listed. For $i\ge 1$, we report a point  
$p_i$ such that $p_i.x\ge p.x$ for any $p\in Q_{i-1}\cap S$, where 
$Q_0=Q$ and $Q_j=[a,p_{i}.x]\times [p_i.y,d]$ for $j\ge 1$.
 Our reporting procedure is completed when $Q_i\cap S=\emptyset$.
Clearly, finding a point $p_i$ or determining that no such $p_i$ 
exists is equivalent to answering a range successor query for 
$Q_{i-1}$. Thus we can find $k$ maximal points in 
$O(kg(n))$ time using an $O(nf(n))$ space data structure, 
where $g(n)$ and $f(n)$ are again defined as in Corollary~\ref{cor:succind}.

A point $p\in S$ is rectangularly visible from a point $q$ 
if $Q_{pq}\cap S=\emptyset$, where $Q_{pq}$ is the rectangle 
with points $p$ and $q$ at its opposite corners. 
In the rectangle visibility problem, we must determine all 
points $p\in S$ that are visible from a query point $q$.  
Rectangular visibility problem is equivalent to 
finding maximal points in $Q\cap S$ for $Q=[0,q.x]\times [0,q.y]$.
Hence, we can find points rectangularly visible from 
$q$ in $O(kg(n))$ time using an $O(nf(n))$ space data structure. 

\bibliographystyle{abbrv}
\bibliography{sorted-rep}

\end{document}